\documentclass[twoside,11pt]{article}

\usepackage{jmlr2e}
\usepackage{amsmath, amssymb}
\usepackage{verbatim}
\usepackage{multirow}
\usepackage{color}
\usepackage[hang,small,bf]{caption}
\usepackage{algorithm,algorithmic}

\newcommand{\SL}{\text{SL}}
\newcommand{\MS}{\text{MS}}

\ShortHeadings{Characterizing Properties for Q-Clustering}{Bosagh Zadeh and Carlsson}
\firstpageno{1}

\begin{document}

\title{Characterizing Properties for Q-Clustering}
       
\author{\name Reza Bosagh Zadeh \email rezab@stanford.edu\\
   \addr Institute for Computational and Mathematical Engineering \\
    Stanford University
    \AND
\name Gunnar Carlsson \email gunnar@math.stanford.edu\\
   \addr Department of Mathematics\\
    Stanford University
}

%\editor{}

\maketitle

\begin{abstract}%   <- trailing '%' for backward compatibility of .sty file
We uniquely characterize two members of the Q-Clustering family in an axiomatic framework. 
We introduce properties that use known tree constructions for the purpose of characterization.
To characterize the Max-Sum
clustering algorithm, we use the Gomory-Hu construction, and to characterize Single-Linkage, we
use the Maximum Spanning Tree. 
Although at first glance it seems these properties are `obviously' all that are necessary to characterize Max-Sum and Single-Linkage, we show that this is not the case, by investigating how subsets of properties interact. We conclude by proposing additions to
the taxonomy of clustering paradigms currently in use.
\end{abstract}


\section{Introduction}

Clustering is a ubiquitous task in unsupervised learning, finding application in a large array of fields other than Computer Science.
Any task where one seeks to group together similar objects and separate dissimilar objects can be thought of as `Clustering'.
Although the problem as defined leaves open a great deal of interpretation \citep{avrimthoughts, ulrikethoughts}, in this
paper we concretely focus on a large class of objectives which can be optimized using Queyranne's algorithm for submodular optimization \citep{queyranne,rizzi}.

This class of objectives for Q-Clustering was introduced and used in \cite{qclustering}, which we analyze from an axiomatic perspective. Recently there has been a significant amount of work on clustering axioms and properties \citep{rita,bosagh2009,gunnar}, 
and we build in this direction to provide properties for the class of algorithms which can be expressed as optima of
Queyranne's optimization algorithm.

\cite{qclustering} considered 2 objectives to which we add and analyze a new objective. The two objectives
in the original Q-clustering paper were Single-Linkage and Minimum Description Length (MDL).
To these we add the Max-Sum objective, which seeks to maximize the sum of similarities inside clusters. 
Queyranne's algorithm can optimize the Single-Linkage criterion perfectly for $k$ clusters, but its
perfect recovery is limited to only 2 clusters for Max-Sum and MDL. This is not surprising, since optimizing
Max-Sum and MDL perfectly for $k>2$ is NP-hard. However, there exist factor 2 approximation algorithms
which proceed by cutting the $k-1$ most expensive edges of the Gomory-Hu tree \citep{gomoryhu} associated with 
MDL and Max-Sum. Gomory-Hu Trees exist in general for any symmetric submodular function. Furthermore, the
natural tree construction for Single-Linkage (SL) is the maximum spanning tree (MST), since the SL criterion
can be optimized by cutting the $k-1$ most expensive edges of the MST.

We focus on the underlying principles behind the Q-Clustering family of objectives, and arrive at the conclusion that
they are uniquely characterized by cutting tree edges from well known tree constructions. Furthermore, we prove that
they are the \textit{unique} clustering functions using the mentioned tree constructions. These constructions are defined formally in section \ref{formal}. We show that Max-Sum and Single-Linkage are axiomatically identical, except for which of the above trees they use in the course of their operation. Furthermore, we show that they are the \textit{only} clustering functions determined by those trees.

The three objectives we consider are quite different. The Single-Linkage objective has the advantage that since we are only comparing similarities,
the similarity measure can be an ordered set (arithmetic operations on the
similarities need not be defined).We formalize the property that allows this criterion to only depend 
on the rank ordering of the similarities, so it is
insensitive to monotone transformations of the similarities. This flexibility is useful in applications where only rankings
are available, for example user studies in which humans only provide rankings instead of actual similarity scores.
Unfortunately Single-Linkage is also very sensitive to outliers since all it takes to merge two clusters is a 
\textit{single path} between the two clusters where all similarities are above a threshold. Since there are generally many paths between
two clusters, this is a stringent requirement on the designers of the similarity matrix. The natural tree construction
associated with the Single-Linkage Criterion is the Maximum Spanning Tree (MST) defined formally in the Section \ref{formal}.

The second objective we consider is the Max-Sum objective, which doesn't suffer from
the single-path problem of singe-linkage, but has the tendency to create a single large
cluster, since having more edges in the objective summation is usually better than fewer. The
tree construction we use for this is the Gomory-Hu tree, explained formally in Section \ref{formal}.

The third objective we consider is probabilistic in nature and is based on the Minimum
Description Length principle. We are given a distribution for each data item, and we attempt to find 
clusters so that describing or encoding the clusters
(separately) can be done using as few bits as possible. It is known that the problem
of finding the optimal clusterings minimizing the description length is equivalent to the
problem of minimizing a symmetric submodular function \citep{qclustering}. The tree construction associated with this objective
is the cut-equivalent tree associated with the symmetric submodular MDL function. We do not provide a
a uniqueness theorem for MDL since it is not clear how to define some our axioms in this setting.

It is important to note that we make a stark contrast between \textit{axioms} and \textit{properties}. Although
they both restrict the class of partitioning functions, we expect axioms to appeal to our intuition about clustering,
whereas properties are simply restrictions on the class of partitioning functions, they need not appeal to any
intuition about clustering.

\subsection{Previous Work}

In addition to submodular optimization algorithm by \cite{queyranne} generalized by \cite{rizzi}, our formal framework is based on Kleinberg's \citep{klein1}. We also adopt two of the three axioms proposed in that paper, Consistency and Scale Invariance. We replace the Richness axiom of \cite{klein1} by its version for the case of fixed number of clusters, $k$-Richness.

An axiomatic characterization of Single-Linkage is available in \cite{bosagh2009}, but they do not provide any characterization of Max-Sum or submodular objectives. \cite{flake} present some algorithms which use Minimum Cut Trees for the purpose of Clustering, but they do not discuss axioms or uniqueness theorems. Another line of attack for characterizing clustering methods is using tools from Topology explored in \cite{carlsson2010characterization} and \cite{carlsson2008persistent}. Instead of fixing $k$, there are other ways of circumventing Kleinberg's impossibility result analyzed in \cite{ackerman09}, and a characterization of the class of hierarchical clustering functions is given in \cite{ackerman10}. Submodular objectives for clustering are explored in \cite{steffi} and \cite{qclustering}.

Finally, \cite{BBV08} and \cite{awasthisupervised} present a framework which assumes that, given some data, 
a target clustering is achieved by interacting with a teacher. Interacting with a teacher is a departure from unsupervised learning.

\subsection{Formal Preliminaries} \label{formal}

A partitioning function acts on a set $S$ of $n \ge 2$ points, and pairwise similarities among the points in $S$. The points in $S$ are not assumed to belong to any specific set; the pairwise similarities are the only data the partitioning function has about them. Since we wish to deal with point sets that do not necessarily belong to a specific set, we identify the points with the set $S = \{1, 2, . . . , n\}$. We can then define a similarity function to be any function $s : S \times S \rightarrow \mathbb{R}^+$  such that for distinct $i, j \in S$, we have $s(i, j) > 0$, i.e. $s$ must be positive symmetric, but there is no requirement of triangle inequality.

Sometimes we write $s = \langle e_1, e_2, \ldots, e_{n \choose 2} \rangle$ to mean the set of edges that exist between all pairs of $n$ points. This list is \textit{always} ordered by \textbf{decreasing} similarity. $w(e)$ is the weight of edge $e$ which connects some two points $i, j$. So $w(e) = s(i, j)$.

A partitioning function is a function $F$ that takes a similarity function $s$ on $S \times S$ and returns a $k$-partitioning of $S$. A $k$-partitioning of $S$ is a collection of $k$ non-empty disjoint subsets of $S$ whose union is $S$. The sets in $F(s)$ will be called its \textit{clusters}. Two partitioning functions are equivalent if and only if they output the same partitioning on all values of $s$ - i.e. functionally equivalent. In the next section we define several trees, each associated with a particular graph $G$.

A natural representation for a similarity function $s = \langle e_1, e_2, \ldots, e_{n \choose 2} \rangle$ is a complete weighted graph $G_s$,
whose $n$ nodes correspond to our objects, and edges correspond to similarity scores assigned by $s$.  Note that when there is no ambiguity about which similarity function is being used, we drop the subscript $s$ and simply use $G$. Also, when it is natural to think of the complete graph for the similarity function, we simply use $s$ to refer to the graph, instead of $G_s$. Now some trees will be associated with a particular $s$.

\subsubsection{Minimum Cut Tree}
\label{mctdesc}
Let $G=(V, E,s)$ be any arbitrary weighted undirected connected graph with $|V|=n$.
Two disjoint subsets $A$ and $B$ of $V$ that also cover $V$ define a cut
in $G$. The sum of the weights of the edges crossing the cut defines the cut value. For two nodes $s,t$, a minimum $s$-$t$ cut is a cut of minimum value that separates nodes $s$ and $t$. For an integer $k$, a minimum $k$-cut is a cut of minimum value that  leaves exactly $k$ connected components.

The Minimum Cut Tree of $G$ - call it MCT($G$) - is a tree which has the same nodes as $G$, but a different set of edges and weights. The edges of the Minimum Cut
Tree of $G$ satisfy the following two properties.

\begin{itemize}
	\item For every two nodes $s, t \in V$, the minimum cut in $G$ that separates these points is given by cutting the smallest edge on the unique path between $s$ and $t$ in MCT($G$).
	\item For every two nodes $s, t \in V$, the weight of the smallest edge on the unique path (in MCT($G$)) connecting them equals the size of a minimum $s$-$t$ cut of $G$.
\end{itemize}

Note that MCT($G$) is \textit{not} a subgraph of $G$ as it has different edges and weight function. One may ask if such trees always exist, and indeed for every undirected graph, there always exists a min-cut tree, and they were initially introduced by \cite{gomoryhu}.

Minimum Cut Trees are not always unique, so we define a canonical MCT function which fixes a particular
ordering on pairs of points (the lexicographical ordering), and uses the algorithm defined in \citep{gomoryhu} and outlined momentarily. Under these conditions, the output of the Gomory-Hu algorithm will be deterministic (thus unique) and we will denote its value MCT($G$).

Given an input graph $G = (V, E)$ the Gomory-Hu algorithm maintains a partition of $V$, $(S_1, S_2, \ldots, S_t)$ and a spanning tree $T$ on the vertex set
$\{S_1, \ldots, S_t\}$. Let $w'$ be the function assigning weights to the edges of $T$. Tree $T$ satisfies the following invariant. \textbf{Invariant:} For any edge $(S_i, S_j)$ in $T$ there are vertices $a$ and $b$ in $S_i$ and $S_j$ respectively, such that $w'(S_i, S_j) = f(a, b)$, where $f$ denotes the minimum cut/maximum flow function, and the cut defined by edge $(S_i, S_j)$ is a minimum $a$-$b$ cut in $G$. This invariant is maintained for $n-1$ steps, after which $T$ is a Gomory-Hu tree of $G$ \citep{vaziranibook}. Each step involves a call to a standard $s\text{-}t$ cut Min-Cut/Max-Flow algorithm.

\subsubsection{Submodular Functions}
Fix a finite set $S$. Submodularity is a property enforced on functions mapping subsets of $S$ set to the reals.  Intuitively, a submodular function over the powerset demonstrates ``diminishing returns".  From a practical perspective, there exist polynomial time algorithms for minimizing submodular functions. In this sense they are the discrete analog of convex functions \citep{awesomesub}.
Let $S$ be a finite set.  A function $f\colon 2^S \to \mathbb{R}$ is submodular iff for all subsets $A$ and $B$ of $S$ we have,
$$f(A)+f(B) \geq f(A \cap B) + f(A \cup B)$$
Furthermore, a symmetric function is one for which $f(S \setminus A) = f(A)$. An important note is that Gomory-Hu trees work and can be generalized because the cut function is both submodular and symmetric. Any submodular symmetric function will induce a Gomory-Hu tree \citep{schrijver2003combinatorial}. An example of a symmetric submodular function is mutual information between two sets of random variables $X$ and $Y$,
$$ I(X;Y) = \int_Y \int_X p(x,y) \log{ \left( \frac{p(x,y)}{p_1(x)\,p_2(y)} \right) } \; dx \,dy,$$
 where $p(x,y)$ is the joint probability distribution function of $X$ and $Y$, and $p_1(x)$ and $p_2(y)$ are the marginal probability distribution functions of $X$ and $Y$ respectively. For a deeper discussion and proofs of submodularity, see \cite{toshev2010submodular}.

\subsubsection{Maximum Spanning Tree}

Given a connected, undirected weighted graph $G$, a spanning tree of $G$ is a subgraph which is a tree and connects all the vertices. A single graph can have many different spanning trees. The weight of a spanning tree is computed as the sum of the weights of the edges in the spanning tree. A Maximum Spanning Tree (MST) is then a spanning tree with weight greater than or equal to every other spanning tree.

Similar to Minimum Cut Trees, MSTs also have a rich history. They can be computed efficiently by Kruskal's algorithm. Maximum Spanning Trees are not always unique, but in the case that we fix an edge ordering, it is well known that they are unique for a particular $G$. We denote the canonical Maximum Spanning Tree of a graph as MST($G$).

\subsubsection{Single-Linkage}

Single-Linkage is the clustering function which starts with all points in singleton clusters, and successively merges clusters until only
$k$ clusters are left. The similarity of two clusters is the similarity of the two most similar points inside differing clusters. 
\begin{comment}
Single-Linkage is made formal
in algorithm \ref{alg:SL}.

\algsetup{indent=2em}
\newcommand{\SingleLinkage}{\ensuremath{\mbox{\sc Single-Linkage}}}
\begin{algorithm}[h!]
\caption{$\SingleLinkage(s, k)$}\label{alg:SL}
\begin{algorithmic} [4]
\STATE Input: $s = \langle e_1, e_2, \ldots, e_{n \choose 2} \rangle$.
\STATE Output: The Single-Linkage $k$-partitioning.
\medskip
\STATE $\Gamma \leftarrow \{ \{1\},\{2\}, \ldots, \{n\} \}$
\STATE $i \leftarrow 1$
\WHILE{$|\Gamma| > k$}
	\STATE let $x$, $y$ be the two ends of $e_i$
	\STATE let $c_x \in \Gamma$, $c_y \in \Gamma$ be the clusters of $x$ and $y$
	\IF{$c_x \neq c_y$} \label{special}
		\STATE Merge $c_x$ and $c_y$
		\STATE $\Gamma \leftarrow (\Gamma \backslash c_x , c_y) \cup \{ c_x \cup c_y \} $
	\ENDIF
	\STATE $i \leftarrow i+1$
\ENDWHILE
\STATE Output $\Gamma$
\end{algorithmic}
\end{algorithm}
\end{comment}

Another way to compute the Single-Linkage $k$-partitioning is to cut the $k-1$ smallest edges of the Maximum Spanning Tree of $G_s$ \citep{mstsl}. It should be noted that the behavior of Single-Linkage is robust against small fluctuations in the weight of the edges in $s$, so long as the order of edges does not change, which can be readily seen from Kruskal's algorithm. %This can be seen readily from algorithm \ref{alg:SL}.

\subsubsection{Max-Sum}

The objective of Max-Sum is to maximize $\Lambda_s(\Gamma) =  \sum_{c \in \Gamma} \sum_{i,j \in c}s(i, j)$ over all $k$-partitionings $\Gamma = \{A,B\}$. 
Finding the optimal partitioning is NP-hard for $k > 2$. However, for $k=2$ finding the optimal Max-Sum 2-partitioning is the same as finding the global  minimum cut and thus poly-time computable. Finding the overall minimum cut is equivalent to cutting the smallest edge of the Minimum Cut Tree, so for $k=2$ Max-Sum can be reinterpreted as the algorithm which cuts the smallest edge of the Minimum Cut Tree.

Since it is not computationally feasible to optimize the above objective function, we define the following approximation algorithm which has a guaranteed approximation factor $2-2/k$ \citep{vaziranibook}: simply iteratively find and remove the global minimum cut until exactly $k$ connected components remain. This algorithm is called the ``MaxCut" clustering function throughout.

\begin{comment}
\algsetup{indent=2em}
\newcommand{\MaxSum}{\ensuremath{\mbox{\sc MaxSum}}}
\begin{algorithm}[h!]
\caption{$\MaxSum(s, k)$}\label{alg:MS}
\begin{algorithmic} [4]
\STATE Input: $s = \langle e_1, e_2, \ldots, e_{n \choose 2} \rangle$.
\STATE Output: The Max-Sum $k$-partitioning.
\medskip
\STATE $G \leftarrow G_s$
\STATE $i \leftarrow 1$
\WHILE{$i < k$}
	\STATE Let C be the global minimum cut in any component of $G$
	\STATE remove C from G
	\STATE $i \leftarrow i+1$
\ENDWHILE
\STATE Output the connected components of $G$
\end{algorithmic}
\end{algorithm}
\end{comment}p

\section{Uniqueness Results}

\subsection{Axioms}

Now in an effort to distinguish clustering functions from partitioning functions, we review some \textit{axioms} (not properties) that one may like a clustering function to satisfy. Here is the first one. If $s$ is a similarity function, then define $\alpha \cdot s$ to be the same function with all similarities multiplied by $\alpha$.
\begin{center}
\textsc{Scale-Invariance.} \textit{For any similarity function $s$, $1 \leq k \leq n$, and scalar $\alpha > 0$, we have $F(s,k) = F(\alpha \cdot s,k)$}
\end{center}

This axiom simply requires the function to be immune to stretching or shrinking the data points linearly. It effectively disallows clustering functions to be sensitive to changes in units of measurement - which is desirable. We would like clustering functions to not have any predefined hard-coded similarity values in their decision process.

The next axiom ensures that the clustering function is ``rich" and not crippled in types of partitioning it could output. For a fixed $S$, Let Range($F(\bullet)$) be the set of all possible outputs while varying $s$.
\begin{center}
\textsc{$k$-Richess.} \textit{Range($F(\bullet, k)$) is equal to the set of all $k$-partitionings of $S$}
\end{center}

In other words, if we are given a set of points such that all we know about the points are pairwise similarities, then for any partitioning $\Gamma$, there should exist a $s$ such that $F(s) = \Gamma$. By varying similarities amongst points, we should be able to obtain all possible $k$-partitionings.

The next axiom is more subtle and was initially introduced in \cite{klein1}, along with richness. We call a partitioning function ``consistent" if it satisfies the following: when we increase similarities between points in the same cluster and decrease similarities between between points in different clusters, we get the same result. Formally, we say that $s'$ is a $\Gamma$-\textit{transformation} of $s$ if (a) for all $i,j \in S$ belonging to the same cluster of $\Gamma$, we have $s'(i,j) \ge s(i,j)$; and (b) for all $i,j \in S$ belonging to different clusters of $\Gamma$, we have $s'(i,j) \le s(i,j)$. In other words, $s'$ is a transformation of $s$ such that points inside the same cluster are made more similar and points not inside the same cluster are made less similar.
\begin{center}
\textsc{Consistency.} \textit{Let $s$ be a similarity function, and $s'$ be a $F(s,k)$-transformation of $s$. Then $F(s,k) = F(s',k)$}
\end{center}

In other words, suppose that we run the partitioning function $F$ on $s$ to get back a particular partitioning $\Gamma$. Now, with respect to $\Gamma$, if we shrink in-cluster similarities or expand between-cluster similarities and run $F$ again, we should still get back the same result - namely $\Gamma$.

The difference between these and the axioms defined by \cite{klein1} is that at all times, the partitioning function $F$ is forced to return a fixed number of clusters. If this were not the case, then the above axioms could never be satisfied by any function. In most popular clustering algorithms such as $k$-means, Single-Linkage,  and spectral clustering, the number of clusters to be returned is determined beforehand -- by the human user or other methods -- and passed into the clustering function as a parameter. Note that it is not clear how to extend Consistency to the case where there is a distribution associated with every point, since there
is not an obvious definition of similarity between two points, which is needed to define Consistency in terms of
increasing or decreasing similarities.

\begin{definition}
A Clustering Function is a partitioning function that satisfies Consistency, Scale Invariance, and $k$-Richness.
\end{definition}

We've already introduced two clustering functions.

\begin{theorem}
Single-Linkage and Max-Sum are Clustering functions.
\end{theorem}

Both Single-Linkage and Max-Sum are clustering functions, with proofs available by \cite{bosagh2009}. So a natural question to ask is what types of \textit{properties} (and not axioms) distinguish these two clustering functions from one another? We introduce two new properties that one may not expect all clustering functions to satisfy, but may at times be desirable.

\subsection{Properties}

For ease of notation, let MST$(s)$ be the Maximum Spanning Tree of $G_s$. Similarly, let MCT$(s)$ be the Minimum Cut Tree of $G_s$. Now we are ready to define two distinguishing properties.

\begin{center}
\textsc{MST-Consistency.} \textit{If $s$ and $s'$ are similarity functions such that MST(s) and MST(s') have the same minimum $k$-cut, then $F(s,k) = F(s',k)$}
\end{center}

In other words, a clustering function is MST-Consistent if it makes all its decisions based on the Maximum Spanning Tree. Note that this does not mean the algorithm must optimize any particular objective, just that its decisions are based on the MST. It is important to note that this property includes both the weights on the edges \textit{and} the structure of the MST. Single-Linkage satisfies this property since it cuts the smallest edge $k-1$ edges of the MST. Note that this is a \textit{property}, not an axiom - we don't expect all clustering functions to make their decisions using the Maximum Spanning Tree. Later in this section we show that Single-Linkage is the \textit{only} clustering function that is MST-Consistent.

Similarly define MCT-Consistency identical to MST-Consistency, with MST replaced with MCT.

\begin{center}
\textsc{MCT-Consistency.} \textit{If $s$ and $s'$ are similarity functions such that MCT(s) and MCT(s') have the same minimum $k$-cut, then $F(s,k) = F(s',k)$}
\end{center}

This property forces a clustering function to make all its decisions based on the Minimum Cut Tree. Max-Sum satisfies this property since it always cuts the smallest $k-1$ edges of the minimum cut tree. We show that Max-Sum is the \textit{only} clustering function that is MCT-Consistent.

Notice that both MST-Consistency and MCT-Consistency imply Scale-Invariance; meaning that if a function is either MST-Consistent or MCT-Consistent, then it is also Scale-Invariant. For this reason, whenever a function satisfies \{MCT, MST\}-Consistency, we ignore Scale-Invariance.

\subsection{Uniqueness Theorems}

Shortly, we will be showing that Single-Linkage and Max-Sum are uniquely characterized by MST-Consistency and MCT-Consistency, respectively. Before doing this, we reflect on the relationships between subsets of our axioms and properties. In doing so, we show that MST/MCT-Consistency properties are \textit{not enough} to characterize SL and Max-Sum, as one might think from a shallow glance. 

We will show that if any of the axioms or properties that we use for proving uniqueness were to be missing, then the uniqueness results do not hold. This means that all axioms and properties are really necessary to prove uniqueness of Single-Linkage and Max-Sum.
\begin{theorem} \label{noteasyms}
Consistency, MCT-Consistency, and $k$-Richness are necessary to characterize Max-Sum.
\end{theorem}
\begin{proof}
For each of the mentioned properties, we show that all the other properties and axioms together are not enough to uniquely characterize Max-Sum. To this end, for each of  property, we exhibit an algorithm that acts differently than Max-Sum, and satisfies all the properties except for one. In other words, we show that without each of these properties, the remaining ones do not uniquely characterize Max-Sum.

\textit{Consistency is necessary.} We define the Minimum Cut Tree Cuts family of partitioning functions. As usual, the task is to partition $n$ points into $k$ clusters. Let $\sigma$ be a permutation function for the set of all $k$-partitionings of $S$. A particular member of the MCT cuts family computes the Max-Sum $k$-partitioning, then runs $\sigma$ on the output of Max-Sum. The entire family is obtained by varying the particular permutation used. Since Max-Sum is $k$-Rich, and $\sigma$ is a bijection, then all the members of MCT-Cuts are also $k$-Rich. Since Max-Sum is MCT-Consistent and $\sigma$ does not look at the input, all the members of MCT-Cuts are also MCT-Consistent. However, it is not true that all members MCT-Cuts are Consistent. This is implied by theorem \ref{mainms}, which says the \textit{only} Consistent member of MST-Cuts is Max-Sum itself, i.e. the case that $\sigma$ is the identity permutation.

\textit{MCT-Consistency is necessary.} Consider that Single-Linkage satisfies Consistency, Scale-Invariance, and $k$-Richess, but is obviously not the same function as Max-Sum. Thus MCT-Consistency is necessary.

\textit{$k$-Richness is necessary.} Now consider the Constant clustering function which always returns the first $n-k+1$ elements of $S$ as a single cluster and returns the remaining $k$ as singleton clusters (a singleton is a cluster with a single point in it), making a total of k clusters. Because this function does not look at $s$, it is trivially MST-Consistent, Consistent, and Scale-Invariant. However, it is not $k$-Rich because it always returns some singletons - i.e. we could never reach a $k$-partitioning that has no singletons.
\end{proof}

\begin{theorem} \label{noteasysl}
Consistency, MST-Consistency, and $k$-Richness are necessary to characterize Single-Linkage.
\end{theorem}
\begin{comment}
\begin{proof}
Similar to what was done for Max-Sum, for each of the mentioned properties, we show that all the other properties are not enough to uniquely characterize Single-Linkage.

\textit{Consistency is necessary.} The Maximum Spanning Tree Cuts family of partitioning functions introduced by \cite{bosagh2009} satisfies MST-Consistency, Scale-Invariance, $k$-Richness, but not Consistency. Thus, leaving out Consistency will be detrimental to any uniqueness theorem for Single-Linkage.

\textit{MST-Consistency is necessary.} Consider that Max-Sum satisfies Consistency, Scale-Invariance, and $k$-Richess, but is obviously not the same function as Single-Linkage. Thus MST-Consistency is necessary.

\textit{2-Richness is necessary.} The constant clustering function described in the proof of theorem \ref{noteasyms} is the exhibit for this claim.
\end{proof}
\end{comment}

Proof omitted for space, but it very similar to proof of theorem \ref{noteasyms}. A summary of these results is available in table \ref{cat}. Now that we have seen our properties do not trivially characterize neither Single-Linkage nor Max-Sum, we can move onto proving the uniqueness theorems.

%table -----------------
\begin{table*}
\begin{center}
  \begin{tabular}{ |c|c|c||c|c| }
    \hline
                     			& Consistency &  $k$-Richness 		& MST-Consistency  & MCT-Consistency \\ \hline
      Single-Linkage 		& \checkmark 	& \checkmark 		& \checkmark		& $\times$		 \\ \hline
      Max-Sum			& \checkmark 	& \checkmark	 	& $\times$ 		& \checkmark		 \\ \hline \hline
      MST cuts family		& $\times$	& \checkmark	 	& \checkmark 		& $\times$		 \\ \hline
      MCT cuts family		& $\times$	& \checkmark	 	& $\times$ 		& \checkmark		 \\ \hline
      Constant partitioning   & \checkmark 	& $\times$	 	& \checkmark 		& \checkmark		 \\
    \hline
  \end{tabular}
\caption{\label{cat} Overview of discussed partitioning functions. Even if one were to consider more partitioning functions, as a consequence of theorems \ref{mainsl} and \ref{mainms}, the Single-Linkage and Max-Sum rows are unique amongst all partitioning functions. }
\end{center}
\end{table*}
%table------------------------------

\begin{lemma} \label{swap}
Given a Consistent partitioning function $F$, and a similarity function $s$ with edges in descending order of similarity
$$s = \langle e_1, e_2, \ldots,e_{p},e_{q}, \ldots  e_{n \choose 2} \rangle$$
then for all $k>0$, if $e_p$ and $e_q$ are both inner edges or both outer edges (w.r.t. $F(s, k)$), we have
$$F(\langle e_1, e_2, \ldots,e_{q},e_{p}, \ldots  e_{n \choose 2} \rangle, k) = F(s, k)$$
\end{lemma}
\begin{proof}
In other words, whenever we have two edges of the same type (inner or outer), in neighboring positions in the edge ordering of $s$, we can swap their positions while maintaining the output of $F$. This is true because if both $e_p$ and $e_q$ are outer edges, then we can shrink $e_p$ until $w(e_p) < w(e_q)$ all the while preserving the output of $F$ (by Consistency). Similarly, if both $e_p$ and $e_q$ are inner edges, we can expand $e_q$ until $w(e_p) < w(e_q)$.
\end{proof}

\begin{theorem} \label{mainms}
Max-Sum is the only MCT-Consistent Clustering function.
\end{theorem}
\begin{proof}
Let $F$ be any Consistent, $k$-Rich, MCT-Consistent clustering function, and let $s$ be any similarity function on $n$ points. $k$ is an integer with $1 \leq k \leq n$. We want to show that for all $s,k$, $F(s,k) = \MS(s,k)$, where $\MS(s,k)$ is the result of  Max-Sum on $s$. For this purpose, we introduce the partitioning $\Gamma$ as whatever the output of MS is on $s$, so $\MS(s,k)=\Gamma$. Whenever we say ``inner" or ``outer" edge for this proof, we mean with respect to $\Gamma$.

By $k$-Richness of $F$, there exists an $s_1$ such that $F(s_1,k) = \MS(s,k) = \Gamma$. Now, through a series of transformations that preserve the output of $F$, we transform $s_1$ into $s_2$, then $s_2$ into $s_3$, $\ldots$, until we arrive at $s$.
\begin{enumerate}
	\item By $k$-Richness, we know there exists an $s_1$ such that $F(s_1, k) = \MS(s, k) = \Gamma$.
	\item Using scale-invariance, we can linearly shrink all edges in $s_1$ until they are all smaller than the smallest edge in $s_1$, call the result $s_2$.
	\item Now using consistency, we can expand all inner edges of $s_1$ until they are exactly of the same weight as they appear in $s$. Call the result $s_3$. Thus, $s_3$ has all inner edges set to exactly the same weight as in $s$, but the outer edges of $s_3$ are all smaller than their counterparts in $s$.
	\item Using consistency, we can shrink all outer edges of $s_3$ until the sum of all outer edges of $s_3$ is smaller than the smallest inner edge of $s_3$ (and $s$). Call the result $s_4$.
	\item By lemma \ref{swap} we can reorder all outer edges in $s_4$ until their order among themselves is the same order as they appear in $s$. Call the result $s_5$.	
	\item For this step it helps to review the construction of Gomory-Hu trees outlined in section \ref{mctdesc}. 
Consider that the sum of all outer edges in $s_5$ is less than any single inner edge in $s_5$. Thus cutting all outer edges in $s_5$ is cheaper
than cutting any single inner edge. Furthermore, the removal of these outer edges in $G_{s_5}$ results in $k$ connected components. To construct a Gomory-Hu tree for $G_{s_5}$, we build the tree by maintaining the Gomory-Hu invariant while querying points in differing clusters of $\Gamma$ with each step. The standard $s$-$t$ MinCut algorithm in the Gomory-Hu iteration is guaranteed to return a subset of the outer edges of $s_5$, since the sum of all outer edges has less weight than any individual inner edge. So after the first $k-1$ iterations of Gomory-Hu, the intermediate tree $T$ will have $k-1$ edges and $k$ supernodes, and the weight of the edges of $T$ will be smaller than any inner edge of $s$. We then run Gomory-Hu to completion to obtain the $T_{s_5}$ = MCT($s_5$). We do the same for $s$, and call the result $T_s$.

\item Since $T_s$ and $T_{s_5}$ are trees, their minimum $k$-cut is given by cutting their $k-1$ smallest edges respectively.
However, since both of their $k-1$ lightest edges correspond to cutting the outer edges of $s_5$, then $T_s$ and $T_{s_5}$ have the same minimum $k$-cut. Thus
by MCT-Consistency we can transform $s_5$ to $s$ while maintaining the output of $F$.
	\item Thus we have $F(s_5,k) = F(s,k) = \Gamma$.
\end{enumerate}
We started with any $s$, and showed that $F(s,k) = \Gamma = \MS(s,k)$.
We also know that Max-Sum satisfies all 3 axioms and MCT-Consistency. Thus
it is uniquely characterized.
\end{proof}

\begin{theorem} \label{mainsl}
Single-Linkage is the only MST-Consistent Clustering function.
\end{theorem}
\begin{proof}

Let $F$ be any Consistent, $k$-Rich, MST-Consistent clustering function, and let $s$ be any similarity function on $n$ points. $k$ is an integer with $1 \leq k \leq n$. We want to show that for all $s,k$, $F(s,k) = \SL(s,k)$, where $\SL(s,k)$ is the result of Single-Linkage on $s$. For this purpose, we introduce the partitioning $\Gamma$ as whatever the output of SL is on $s$, so $\SL(s,k)=\Gamma$. Whenever we say ``inner" or ``outer" edge for this proof, we mean with respect to $\Gamma$.

By $k$-Richness of $F$, there exists an $s_1$ such that $F(s_1,k) = \SL(s,k) = \Gamma$. Now, through a series of transformations that preserve the output of $F$, we transform $s_1$ into $s_2$, then $s_2$ into $s_3$, $\ldots$, until we arrive at $s$.
\begin{enumerate}
	\item By $k$-Richness, we know there exists an $s_1$ such that $F(s_1, k) = \SL(s, k) = \Gamma$.
	\item Using scale-invariance, we can linearly shrink all edges in $s_1$ until they are all smaller than the smallest edge in $s_1$, call the result $s_2$.
	\item Now using consistency, we can expand all inner edges of $s_1$ until they are exactly of the same weight as they appear in $s$. Call the result $s_3$. Thus, $s_3$ has all inner edges set to exactly the same weight as in $s$, but the outer edges of $s_3$ are all smaller than their counterparts in $s$.
	\item Using consistency, we can shrink all outer edges of $s_3$ until the sum of all outer edges of $s_3$ is smaller than the smallest inner edge of $s_3$ (and $s$). Call the result $s_4$.
	\item By lemma \ref{swap} we can reorder all outer edges in $s_4$ until their order among themselves is the same order they appear in $s$, we can do the same for the inner edges. Call the result $s_5$.	
	\item By Kruskal's algorithm for MST, $s$ and $s_5$ have the same MST, call them $T_s$ and $T_{s_5}$. Since $T_s$ and $T_{s_5}$ are trees, their minimum $k$-cut is given by cutting their $k-1$ smallest edges respectively.
However, since both of their $k-1$ lightest edges correspond to cutting the outer edges of $s_5$, then $T_s$ and $T_{s_5}$ have the same minimum $k$-cut. Thus
by MST-Consistency we can transform $s_5$ to $s$ while maintaining the output of $F$.
	\item Thus we have $F(s_5,k) = F(s,k) = \Gamma$.
\end{enumerate}
We started with any $s$, and showed that $F(s,k) = \Gamma = \SL(s,k)$.
We also know that Single-Linkage satisfies all 3 axioms \citep{bosagh2009} and MST-Consistency. Thus
it is uniquely characterized.
\end{proof}

\section{Conclusions \& Future directions}

In this paper we have characterized two clustering algorithms which are usually treated separately with differing motivating principles. Using our framework, one can aim to build a suite of abstract properties of clustering that will induce a taxonomy of clustering paradigms. Such a taxonomy should serve to help utilize prior domain knowledge to allow educated choice of a clustering method that is appropriate for a given clustering task.

The chief contributions of this paper are the characterizations of clustering through associated tree constructions. The tree construction relevant for Max-Sum turned out to be the Gomory-Hu tree, and the tree construction for Single-Linkage was the Maximum Spanning Tree. It is expected that the general Gomory-Hu tree construction for \textit{any} submodular function can also be used for characterizing other objectives, but it is not immediately clear how to change axioms such as Consistency to achieve this and we leave it for future work.

Our contribution provides insight into the connection between Submodularity, Single-Linkage, and Max-Sum clustering functions. For the latter two, we show they are axiomatically identical except for one property. By considering the listing in table \ref{cat}, we demonstrate the type of desired taxonomy of clustering functions based on the properties each satisfies. 

Although at first glance it seems these properties are `obviously' all that are necessary to characterize Max-Sum and Single-Linkage, we show that this is not the case, by way of Theorems \ref{noteasyms} and \ref{noteasysl}. To investigate the ramifications of algorithms which satisfy only a subset of our properties, we introduced the MCT Cuts family of partitioning functions, of which Max-Sum is the only Consistent member. The uniqueness theorems came about as a result of forcing functions to focus on Minimum/Maximum Cut/Spanning Trees.

% Acknowledgements should go at the end, before appendices and references

\acks{We would like to acknowledge Shai Ben-David, Stefanie Jegelka, Avrim Blum, and Ashish Goel for very valuable discussions.}


\newpage

\vskip 0.2in
\bibliography{jmlr2012}

\begin{thebibliography}{23}
\providecommand{\natexlab}[1]{#1}
\providecommand{\url}[1]{\texttt{#1}}
\expandafter\ifx\csname urlstyle\endcsname\relax
  \providecommand{\doi}[1]{doi: #1}\else
  \providecommand{\doi}{doi: \begingroup \urlstyle{rm}\Url}\fi

\bibitem[Ackerman and Ben-David(2008)]{rita}
Margarita Ackerman and Shai Ben-David.
\newblock {Measures of Clustering Quality: A Working Set of Axioms for
  Clustering}.
\newblock In \emph{Advances in Neural Information Processing Systems:
  Proceedings of the 2008 Conference}, 2008.

\bibitem[Ackerman and Ben-David(2009)]{ackerman09}
Margarita Ackerman and Shai Ben-David.
\newblock {Clusterability: A theoretical study}.
\newblock \emph{Proceedings of AISTATS-09, JMLR: W\&CP}, 5:\penalty0 1--8,
  2009.

\bibitem[Ackerman et~al.(2010)Ackerman, Ben-David, and Loker]{ackerman10}
Margarita Ackerman, Shai Ben-David, and David Loker.
\newblock {Characterization of Linkage-based Clustering}.
\newblock \emph{COLT 2010}, 2010.

\bibitem[Awasthi and Zadeh(2010)]{awasthisupervised}
Pranjal Awasthi and Reza~Bosagh Zadeh.
\newblock {Supervised Clustering}.
\newblock \emph{Neural Information Processing Systems (NIPS 2010)}, 2010.

\bibitem[Balcan et~al.(2008)Balcan, Blum, and Vempala]{BBV08}
M.-F. Balcan, A.~Blum, and S.~Vempala.
\newblock A discriminative framework for clustering via similarity functions.
\newblock In \emph{Proceedings of the 40th ACM Symposium on Theory of
  Computing}, 2008.

\bibitem[Blum(2009)]{avrimthoughts}
Avrim. Blum.
\newblock Thoughts on clustering.
\newblock In \emph{NIPS Workshop on Clustering Theory}, 2009.

\bibitem[Carlsson and M{\'e}moli(2008)]{carlsson2008persistent}
Gunnar Carlsson and Facundo M{\'e}moli.
\newblock {Persistent Clustering and a Theorem of J}.
\newblock \emph{Kleinberg. ArXiv e-prints}, 2008.

\bibitem[Carlsson and
  M{\'e}moli(2010{\natexlab{a}})]{carlsson2010characterization}
Gunnar Carlsson and Facundo M{\'e}moli.
\newblock {Characterization, stability and convergence of hierarchical
  clustering methods}.
\newblock \emph{The Journal of Machine Learning Research}, 99:\penalty0
  1425--1470, 2010{\natexlab{a}}.
\newblock ISSN 1532-4435.

\bibitem[Carlsson and M{\'e}moli(2010{\natexlab{b}})]{gunnar}
Gunnar Carlsson and Facundo M{\'e}moli.
\newblock Characterization, stability and convergence of hierarchical
  clustering methods.
\newblock \emph{The Journal of Machine Learning Research}, 99:\penalty0
  1425--1470, 2010{\natexlab{b}}.

\bibitem[Flake et~al.(2004)Flake, Tarjan, and Tsioutsiouliklis]{flake}
G.W. Flake, R.E. Tarjan, and K.~Tsioutsiouliklis.
\newblock {Graph clustering and minimum cut trees}.
\newblock \emph{Internet Mathematics}, 1\penalty0 (4):\penalty0 385--408, 2004.

\bibitem[Gomory and Hu(1961)]{gomoryhu}
RE~Gomory and TC~Hu.
\newblock {Multi-terminal network flows}.
\newblock \emph{Journal of the Society for Industrial and Applied Mathematics},
  pages 551--570, 1961.

\bibitem[Gower and Ross(1969)]{mstsl}
J.C. Gower and GJS Ross.
\newblock {Minimum spanning trees and single linkage cluster analysis}.
\newblock \emph{Applied Statistics}, 18\penalty0 (1):\penalty0 54--64, 1969.

\bibitem[Guyon et~al.(2009)Guyon, Von~Luxburg, and Williamson]{ulrikethoughts}
Isabelle Guyon, Ulrike Von~Luxburg, and Rob~.C. Williamson.
\newblock Clustering: Science or art.
\newblock In \emph{NIPS 2009 Workshop on Clustering Theory}, 2009.

\bibitem[Jegelka and Bilmes(2010)]{steffi}
Stefanie Jegelka and Jeff Bilmes.
\newblock Cooperative cuts: Graph cuts with submodular edge weights.
\newblock Technical report, Technical Report, 2010.

\bibitem[Kleinberg(2003)]{klein1}
Jon Kleinberg.
\newblock {An Impossibility Theorem for Clustering}.
\newblock In \emph{Advances in Neural Information Processing Systems 15:
  Proceedings of the 2002 Conference}. MIT Press, 2003.

\bibitem[Lov{\'a}sz(1983)]{awesomesub}
Laszlo Lov{\'a}sz.
\newblock {Submodular functions and convexity}.
\newblock \emph{Mathematical programming: the state of the art}, pages
  235--257, 1983.

\bibitem[Narasimhan et~al.(2006)Narasimhan, Jojic, and Bilmes]{qclustering}
Mukund Narasimhan, Nebojsa Jojic, and Jeff Bilmes.
\newblock Q-clustering.
\newblock \emph{Advances in Neural Information Processing Systems},
  18:\penalty0 979, 2006.

\bibitem[Queyranne(1998)]{queyranne}
Maurice Queyranne.
\newblock Minimizing symmetric submodular functions.
\newblock \emph{Mathematical Programming}, 82\penalty0 (1):\penalty0 3--12,
  1998.

\bibitem[Rizzi(2000)]{rizzi}
Romeo Rizzi.
\newblock Note--on minimizing symmetric set functions.
\newblock \emph{Combinatorica}, 20\penalty0 (3):\penalty0 445--450, 2000.

\bibitem[Schrijver(2003)]{schrijver2003combinatorial}
A.~Schrijver.
\newblock \emph{{Combinatorial optimization: polyhedra and efficiency}}.
\newblock Springer Verlag, 2003.
\newblock ISBN 3540443894.

\bibitem[Toshev(2010)]{toshev2010submodular}
A~Toshev.
\newblock Submodular function minimization.
\newblock \emph{University of Pennsylvania, Philadelphia, Tech. Rep}, 2010.

\bibitem[Vazirani(2001)]{vaziranibook}
Vijay Vazirani.
\newblock \emph{{Approximation algorithms}}.
\newblock Springer Verlag, 2001.
\newblock ISBN 3540653678.

\bibitem[Zadeh and Ben-David(2009)]{bosagh2009}
Reza~Bosagh Zadeh and Shai Ben-David.
\newblock {A Uniqueness Theorem for Clustering}.
\newblock In \emph{Conference on Uncertainty in Artificial Intelligence (UAI
  2009)}, 2009.

\end{thebibliography}

\end{document}